\documentclass[article,oneside,12pt,a4paper]{memoir}

\usepackage[utf8]{inputenc}
\usepackage[T1]{fontenc}
\usepackage[shrink=10,stretch=10]{microtype}
\usepackage{amsmath,amssymb,amsthm}
\usepackage[british]{babel}
\usepackage[svgnames]{xcolor}
\usepackage{graphicx}
\usepackage{booktabs}
\usepackage{calc}
\usepackage[hidelinks]{hyperref}
\usepackage{tikz}
\usetikzlibrary{patterns}

\usepackage{enumitem}
\setlist{
  topsep=1ex,
  itemsep=-0.5ex,
}

\setlrmarginsandblock{3.6cm}{3.6cm}{*}
\setulmarginsandblock{3cm}{3cm}{*}
\checkandfixthelayout

\makeatletter
\def\@endtheorem{\endtrivlist}
\makeatother
\newtheorem{theorem}{Theorem}

\newtheorem{proposition}[theorem]{Proposition}
\newtheorem{corollary}[theorem]{Corollary}
\newtheorem{definition}[theorem]{Definition}

\newtheorem{example}{Example}
\newtheorem{algorithm}{Algorithm}

\hypersetup{
    pdftitle={On Steane-Enlargement of Quantum Codes from Cartesian Product Point Sets},
    pdfauthor={René Bødker Christensen and Olav Geil}
}

\captionnamefont{\footnotesize\bfseries}	
\captiontitlefont{\footnotesize}					
\captiondelim{. }
\setsecheadstyle{\bfseries\large}

\counterwithout{section}{chapter}
\setsecnumdepth{subsection}
\setsubsecheadstyle{\bfseries}

\DeclareMathOperator{\ev}{ev}
\DeclareMathOperator{\Span}{Span}
\renewcommand{\vec}[1]{\mathbf{\boldsymbol #1}}
\newcommand{\ph}{\phantom{0}}
\newcommand{\sep}{\unskip,\ }
\newcommand{\orcidID}[1]{
    \unskip\hspace{0.1em}\href{https://orcid.org/#1}{\raisebox{-0.12\height}{\includegraphics[height=0.8em]{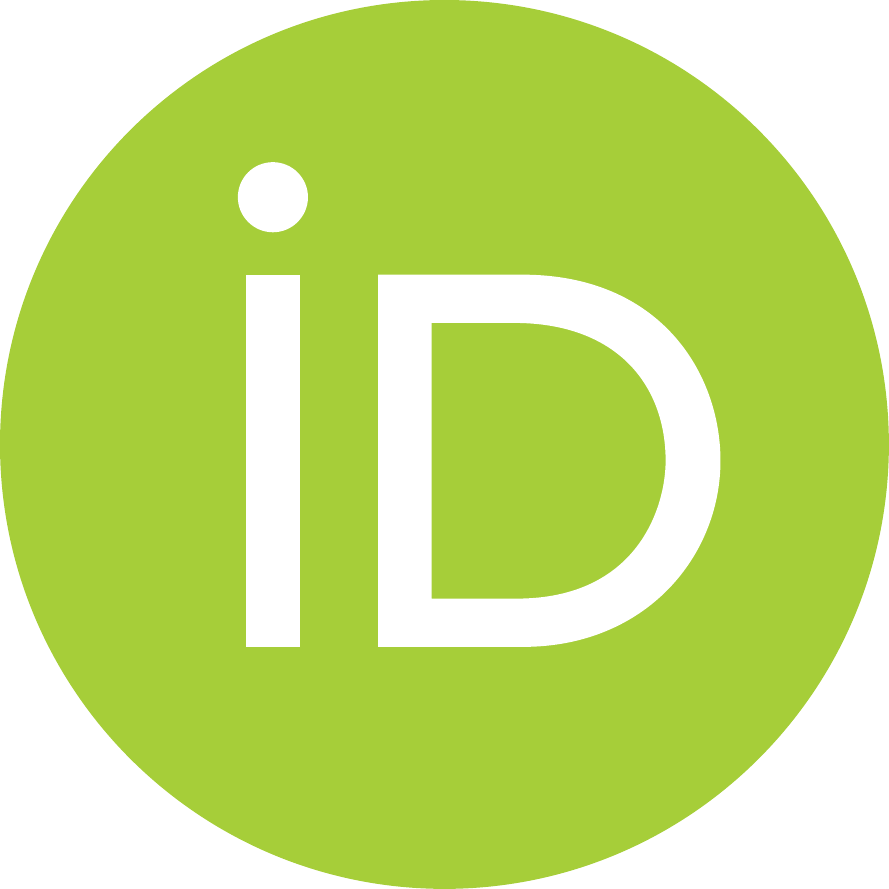}}}
}
\newcommand{\trueMark}{\ensuremath{^\ast}}
\newcommand{\phMark}{\phantom{\trueMark}}
\makeatletter
\newcommand{\gvMark}{\@ifstar{^\dagger}{^\ddagger}}
\makeatother

\begin{document}
\begin{center}
  {\Large\fontfamily{ppl}\fontseries{b}\selectfont On Steane-Enlargement of Quantum Codes\\\smallskip from Cartesian Product Point Sets}\\[1.5em]
  {René Bødker Christensen\orcidID{0000-0002-9209-3739}and Olav Geil\orcidID{0000-0002-9666-3399}}
  
  \bigskip{Department of Mathematical Sciences,
      Aalborg University, Denmark.}
  
  {\texttt{\{rene,olav\}@math.aau.dk}}

  \vglue 2em
  {\small\textbf{Abstract}}\\[0.5em]
  \begin{minipage}{0.9\linewidth}
    \rule{\linewidth}{0.5pt}
    \footnotesize In this work, we study quantum error-correcting
    codes obtained by using Steane-enlargement. We apply this technique to
    certain codes defined from Cartesian products previously considered by
    Galindo et al. in \cite{GGHR18}. We give bounds on the dimension
    increase obtained via enlargement, and additionally give an algorithm
    to compute the true increase. A number of examples of codes are
    provided, and their parameters are compared to relevant codes in the
    literature, which shows that the parameters of the enlarged codes are
    advantageous. Furthermore, comparison with the Gilbert-Varshamov bound
    for stabilizer quantum codes shows that several of the enlarged codes
    match or exceed the parameters promised by the bound.


    \rule{\linewidth}{0.5pt}
  \end{minipage}
\end{center}

\section{Introduction}\label{sec:intro}
In \cite{GGHR18}, Galindo et al. give two constructions of quantum error-correcting codes defined from Cartesian product point sets, and the resulting codes have good parameters. In their work, they consider \emph{asymmetric} quantum codes, meaning that phase-shift and bit-flip errors are not treated equally and both of the corresponding minimal distances $d_z,d_x$ are of interest. Another viewpoint commonly seen in the literature is that no distinction should be made between the two types of errors when assessing the error-correcting capabilities of a quantum error-correcting code. In this setting, only one distance $d=\min\{d_z,d_x\}$ is associated to the quantum code, and the code is called \emph{symmetric}. Clearly, the parameters in the asymmetric setting can be translated into the symmetric setting by ignoring the highest distance. Thus, someone interested in symmetric codes could use the results from \cite{GGHR18}, but discarding the highest distance essentially wastes coding space which could instead be used to increase the dimensions of the codes. In this work, we take an alternative approach and apply Steane-enlargement to that family of codes in order to produce symmetric codes directly. We thereby produce quantum error-correcting codes with good -- sometimes even optimal -- parameters.

The classical codes considered in this work are special cases of what is called \emph{monomial Cartesian codes} in a recent work \cite{LMS19}. In that paper, the authors derive a way to determine if a monomial Cartesian code is self-orthogonal, and use this to construct quantum codes via the CSS-construction. The classical codes used in their construction are, however, different from the ones used in the current paper. In particular, the improved codes considered in this work have the best possible dimension given any designed distance.

\section{Preliminaries}\label{sec:prelim}
In this section, we recall two results on the CSS-construction and Steane-enlargement that allow construction of quantum codes from classical codes. Then we give a description of a family of codes and  the corresponding improved codes, both of which were previously considered in \cite{GGHR18}.
In our analysis, we will rely on the notion of relative distances of nested pairs of classical linear codes. Thus, recall that for codes $\mathcal{C}_2\subsetneq\mathcal{C}_1$ their relative distance is defined as
\begin{equation*}
  d(\mathcal{C}_1,\mathcal{C}_2)=\min\{w_H(\vec{c})\mid\vec{c}\in\mathcal{C}_1\setminus\mathcal{C}_2\}
\end{equation*}
where $w_H$ denotes the usual Hamming weight. In general, however, the relative distance is difficult to determine, and the bound $d(\mathcal{C}_1,\mathcal{C}_2)\geq d(\mathcal{C}_1)$ is commonly used instead.

\subsection{The CSS-construction and Steane-enlargement}\label{sec:cssAndEnlarge}
One way to construct quantum error-correcting codes is by using the so-called CSS-construction \cite{CalderbankShor,Steane96} named after Calderbank, Shor, and Steane. The original construction uses a self-orthogonal classical linear code to construct a symmetric quantum error-correcting code.
\begin{theorem}\label{thm:cssSelfOrth}
  If the $[n,k,d]$ linear code $\mathcal{C}\subseteq\mathbb{F}_q^n$ contains its Euclidean dual, then an
  \begin{equation*}
    [[n,2k-n,d]]_q
  \end{equation*}
  symmetric quantum code exists.
\end{theorem}
Steane \cite{Steane99} proposed a variation on this procedure which in some cases allows an increase in dimension compared to the corresponding CSS-code but without reducing the minimal distance. Below, we state the $q$-ary generalization of this procedure, which may be found in \cite{Hamada,qArySteane}.
\begin{theorem}\label{thm:qArySteane}
  Consider a linear $[n,k]$ code $\mathcal{C}\subseteq\mathbb{F}_q^n$ that contains its Euclidean dual $\mathcal{C}^\perp$. If $\mathcal{C}'$ is an $[n,k']$ code such that $\mathcal{C}\subsetneq\mathcal{C}'$ and $k'\geq k+2$, then an
  \begin{equation*}
    \Big[\Big[n,k+k'-n,\geq \min\big\{d,\big\lceil(1+\textstyle\frac{1}{q})d'\big\rceil\big\}\Big]\Big]_q
  \end{equation*}
  quantum code exists with $d=d(\mathcal{C},\mathcal{C}'^\perp)$ and $d'=d(\mathcal{C}',\mathcal{C}'^\perp)$.
\end{theorem}
Here we note that if $\mathcal{C}$ and $\mathcal{C}'$ are codes that satisfy the conditions of Theorem~\ref{thm:qArySteane}, then the inclusions $\mathcal{C}'^\perp\subsetneq\mathcal{C}^\perp\subseteq\mathcal{C}\subsetneq\mathcal{C}'$ hold, which implies $d(\mathcal{C}'^\perp)\geq d(\mathcal{C})$. In particular, this means that whenever $d(\mathcal{C}')<d(\mathcal{C})$, it must be the case that $d'=d(\mathcal{C}',\mathcal{C}'^\perp)=d(\mathcal{C}')$. For the specific enlargements considered in Section~\ref{sec:steane-enlarg-impr}, it turns out that this observation allows us to use the usual minimal distances rather than the relative distances while still obtaining the same parameters of the quantum codes.

\subsection{Codes from Cartesian product point sets}
Let $q=p^r$ where $p$ is a prime number, and let $r_1,r_2,\ldots,r_m$ be positive integers such that $r_i\mid r$. Then we have the inclusions $\mathbb{F}_{p^{r_i}}\subseteq\mathbb{F}_q$, and it is possible to consider the Cartesian product $S=\mathbb{F}_{p^{r_1}}\times\mathbb{F}_{p^{r_2}}\times\cdots\times\mathbb{F}_{p^{r_m}}\subseteq\mathbb{F}_q^m$. Now, define the polynomials
\begin{equation*}
  F_i(X_i)=\prod_{\alpha\in\mathbb{F}_{p^{r_i}}}(X_i-\alpha)=X_i^{p^{r_i}}-X_i,
\end{equation*}
and consider the ring $R=\mathbb{F}_q[X_1,X_2,\ldots,X_m]/I$ where
\begin{equation*}
  I=\langle F_1(X_1),F_2(X_2),\ldots,F_m(X_m)\rangle
\end{equation*}
is the vanishing ideal of the $F_i$'s. Letting $n=|S|=\prod_{i=1}^m p^{r_i}$ and $S=\{\vec{\alpha}_1,\vec{\alpha}_2,\ldots,\vec{\alpha}_n\}$, we obtain a vector space homomorphism $\ev\colon R\rightarrow\mathbb{F}_q^n$ given by
\begin{equation*}
  \ev(F+I)=(F(\vec{\alpha}_1),F(\vec{\alpha}_2),\ldots,F(\vec{\alpha}_n))
\end{equation*}
as described in \cite{GGHR18}.
Adopting a vectorized version of their notation, we define for $\vec{r}=(r_1,r_2,\ldots,r_m)$ the set
\begin{equation*}
  \Delta(\vec{r})=\{X^{\vec{a}}\mid \vec{a}\in\mathbb{N}^m,\; 0\leq a_j<p^{r_j},\; j=1,2,\ldots,m\},
\end{equation*}
where we use the multi-index notation $X^{\vec{a}}=X_1^{a_1}X_2^{a_2}\cdots X_m^{a_m}$.
For a subset $L\subseteq\Delta(\vec{r})$, define the code
\begin{equation}\label{eq:codeDef}
  C(L)=\Span_{\mathbb{F}_q}\{\ev(X^{\vec{a}}+I)\mid X^{\vec{a}}\in L\},
\end{equation}
which clearly has length $n$.
To describe the distance of $C(L)$, we use the map $\sigma\colon\Delta(\vec{r})\rightarrow\mathbb{N}$ given by
\begin{equation*}
  \sigma(X^{\vec{a}})=\prod_{j=1}^m(p^{r_j}-a_j).
\end{equation*}
\begin{proposition}\label{prop:dimAndDist}
  Let $C(L)$ be defined as in \eqref{eq:codeDef}. Then $\dim C(L)=|L|$, and
  \begin{equation}\label{eq:minDistBound}
    d(C(L))\geq \min\{\sigma(X^{\vec{a}})\mid X^{\vec{a}}\in L\}
  \end{equation}
  with equality if $X^{\vec{a}}\in L$ implies $X^{\vec{b}}\in L$ for all choices of $b_1\leq a_1, b_2\leq a_2,\ldots,b_m\leq a_m$.
\end{proposition}
\begin{proof}
  The claim about the dimension is for instance shown in the proof of \cite[Thm.\,16]{GGHR18}. The inequality \eqref{eq:minDistBound} can be proved by using the footprint bound as done in \cite[Prop.\,1]{GH01}.

  To see the equality, write $\mathbb{F}_{p^{r_i}}=\{v_1^{(i)},v_2^{(i)},\ldots,v_{p^{r_i}}^{(i)}\}$, let $X^{\vec{a}}\in L$, and observe that the expansion of the polynomial
  \begin{equation*}
    f=\prod_{j=1}^m\prod_{i=1}^{a_j}(X_j-v_i^{(j)})
  \end{equation*}
  contains only monomials $X^{\vec{b}}$ with $\vec{b}$ as described in the proposition, meaning that $\ev(f+I)\in C(L)$. Moreover, it possesses exactly $\prod_{j=1}^m(p^{r_i}-a_j)=\sigma(X^{\vec{a}})$ non-zeros.
\end{proof}
This proposition not only allows us to determine the exact minimal distance of the codes considered in the following section, but more importantly it also enables us to determine certain relative distances when combined with the observations below Theorem~\ref{thm:qArySteane}.

\subsection{Improved codes}\label{subsec:improved-codes}
The information on the minimal distance provided by $\sigma$ leads to improved code constructions in a straightforward manner. By defining
\begin{equation}\label{eq:improvedMonomials}
  L(\delta)=\{X^{\vec{a}}\in\Delta(\vec{r})\mid\sigma(X^{\vec{a}})\geq\delta\}
\end{equation}
the code $C(L(\delta))$ has designed distance $\delta$ by Proposition~\ref{prop:dimAndDist}. In addition, this is the true minimal distance since $\sigma(X^{\vec{b}})\geq\sigma(X^{\vec{a}})$ if $b_1\leq a_1,b_2\leq a_2,\ldots,b_m\leq a_m$.
The dual of $C(L(\delta))$ can be described by studying the map $\mu\colon\Delta(\vec{r})\rightarrow\mathbb{N}$ defined as
\begin{equation*}
  \mu(X^{\vec{a}})=\prod_{j=1}^m(a_j+1).
\end{equation*}
In particular, by letting $L^\perp(\delta)=\{X^{\vec{a}}\in\Delta(\vec{r})\mid\mu(X^{\vec{a}})<\delta\}$ we obtain the following result.
\begin{proposition}
  Let $L(\delta)$ be defined as in \eqref{eq:improvedMonomials}. Then $C(L(\delta))^\perp=C(L^\perp(\delta))$.
\end{proposition}
\begin{proof}
  First, note that $\sigma(X^{\vec{a}})=\mu(X^{\vec{b}})$ for $b_i=p^{r_i}-a_j-1$. This implies that the number of monomials with a given $\sigma$-value $\delta$ is exactly the number of monomials with $\mu$-value $\delta$. As a consequence,
  \begin{align*}
    \dim C(L^\perp(\delta)) &=|\{X^{\vec{a}}\in\Delta(\vec{r})\mid \sigma(X^{\vec{a}})<\delta\}|\\
                            &=n-\dim C(L(\delta))\\
                            &=\dim C(L(\delta))^\perp.
  \end{align*}
  Hence, it suffices to show that $C(L^\perp(\delta))\subseteq C(L(\delta))^\perp$, and we do so by proving that the evaluation of any $X^{\vec{b}}$ with $\mu(X^{\vec{b}})<\delta$ must be in $C(L(\delta))^\perp$.

  Using contraposition, assume that $X^{\vec{b}}\notin C(L(\delta))^\perp$. Then some $\ev(X^{\vec{a}})\in C(L(\delta))$ satisfies $\ev(X^{\vec{a}})\cdot\ev(X^{\vec{b}})\neq 0$. As shown in \cite[Prop.\,1]{GHR15}, this happens if and only if\footnote{In their notation, the situation in consideration has $J=\emptyset$ and $p\mid N_j$ for each $j$} $a_i+b_i>0$ and $a_i+b_i\equiv 0\pmod{p^{r_i}-1}$ holds true for each index $i\in\{1,2,\ldots,m\}$. In other words, we have $a_i+b_i=p^{r_i}-1$ or $a_i+b_i=2(p^{r_i}-1)$. In each case, this implies $p^{r_i}-a_i\leq b_i+1$. In combination with the fact that $\sigma(X^{\vec{a}})\geq\delta$ since $\ev(X^{\vec{a}})\in C(L(\delta))$, we obtain the inequalities
  \begin{equation*}
    \delta\leq\sigma(X^{\vec{a}})=\prod_{i=1}^m(p^{r_i}-a_i)\leq\prod_{i=1}^m(b_i+1)=\mu(X^{\vec{b}}).
  \end{equation*}
  In conclusion, if $\mu(X^{\vec{b}})<\delta$, we have $\ev(X^{\vec{b}})\in C(L(\delta))^\perp$ which proves the proposition by the observations in the beginning of the proof.
\end{proof}

\begin{figure}[b]
  \centering
  \begin{tikzpicture}[scale=0.7]
    \pgfmathsetmacro{\offset}{0.35}
    \tikzset{shade/.style={Gray!40,opacity=0.5}}

    \fill[draw=Gray!50,shade,rounded corners=2pt](-\offset,2-\offset) -- ++(0,2*\offset) -- ++(8+2*\offset,0) -- ++(0,-2-2*\offset) -- ++(-2*\offset,0) -- ++(0,2) -- cycle;
    \draw (0,0) node{$27$} (0,1) node{$18$} (0,2) node{$9$} (1,0) node{$24$} (1,1) node{$16$} (1,2) node{$8$} (2,0) node{$21$} (2,1) node{$14$} (2,2) node{$7$} (3,0) node{$18$} (3,1) node{$12$} (3,2) node{$6$} (4,0) node{$15$} (4,1) node{$10$} (4,2) node{$5$} (5,0) node{$12$} (5,1) node{$8$} (5,2) node{$4$} (6,0) node{$9$} (6,1) node{$6$} (6,2) node{$3$} (7,0) node{$6$} (7,1) node{$4$} (7,2) node{$2$} (8,0) node{$3$} (8,1) node{$2$} (8,2) node{$1$};
  \end{tikzpicture}
  \caption{The values of $\sigma(\Delta(\vec{r}))$ for $p=3$ and $\vec{r}=(2,1)$. The shaded region shows the edges with values $1,2,3=p^{r_1}$ and $1,2,\ldots,9=p^{r_2}$, respectively.}
  \label{fig:example2D}
\end{figure}
\begin{figure}[b]
  \centering
  \begin{tikzpicture}
    \pgfsetzvec{\pgfpoint{-0.5cm}{-0.25cm}}
    \pgfmathsetmacro{\lenA}{5}
    \pgfmathsetmacro{\lenB}{1.5}
    \pgfmathsetmacro{\lenC}{1.8}
    \pgfmathsetmacro{\offset}{0.2}
    \pgfmathsetmacro{\adjA}{0.4} 
    \pgfmathsetmacro{\adjB}{0.1} 
    \tikzset{shade/.style={Gray!40,opacity=0.5}}

    \draw[dashed,gray](0,0,-\lenC) edge (0,0,0) edge (0,\lenB,-\lenC) edge (\lenA,0,-\lenC);

    \fill[shade](0,\lenB,-\lenC) --
        ++(0,0,\offset) --
        ++(0,-\offset,0) --
        ++(\lenA-\offset-\adjA,0,0) --
        ++(\adjA,\offset,0) --
        ++(0,0,-\offset) --
        cycle;
    \fill[shade](\lenA-\offset,0,-\lenC+\offset) --
        ++(\offset,0,0) --
        ++(0,0,-\offset) --
        ++(0,\lenB-\offset,0) --
        ++(0,0,\offset) --
        ++(-\offset,-\adjB,0) --
        cycle;
    \fill[shade](\lenA-\offset,\lenB-\offset,0) --
        ++(\offset,0,0) --
        ++(0,0,-\lenC) --
        ++(0,\offset,0) --
        ++(-\offset,0,0) --
        ++(0,0,\lenC) --
        cycle;

    \draw[Gray!50](\lenA,\lenB-\offset,0) --
    ++(0,0,-\lenC+\offset) --
    ++(0,-\lenB+\offset) --
    ++(-\offset,0,0) --
    ++(0,\lenB-\offset-\adjB,0) ++(0,\adjB,0) ++(-\adjA,0,0) --
    ++(-\lenA+\offset+\adjA,0,0) --
    ++(0,\offset,0) --
    ++(\lenA-\offset,0,0) --
    ++(0,0,\lenC-\offset) --
    ++(0,-\offset,0) --
    ++(\offset,0,0);

    \draw[thick](0,0,0) -- node[below]{$p^{r_1}$}
        ++(\lenA,0,0) -- node[below right=-2pt]{$p^{r_3}$}
        ++(0,0,-\lenC) --
        ++(0,\lenB,0) --
        ++(-\lenA,0,0) -- 
        ++(0,0,\lenC) -- node[left]{$p^{r_2}$}
        cycle;
    \draw[thick](\lenA,\lenB,0) edge (\lenA,0,0) edge (0,\lenB,0) edge (\lenA,\lenB,-\lenC);
  \end{tikzpicture}
  \caption{A sketch of $\Delta(\vec{r})$ in the case $m=3$. As in the 2-dimensional case in Figure~\ref{fig:example2D}, the shaded region shows the edges where the $\sigma$-values are $1,2,\ldots,p^{r_i}$ for each $i$.}
  \label{fig:example3D}
\end{figure}
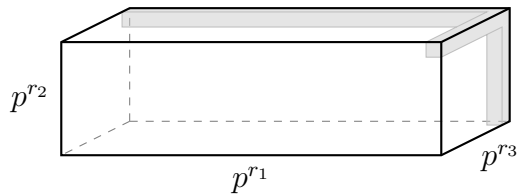

\section{Steane-enlargement of improved codes}\label{sec:steane-enlarg-impr}
We are now ready to apply Steane-enlargement to the codes defined in Section~\ref{subsec:improved-codes}. Our results rely on a simple, but crucial, observation: for each index $i=1,2,\ldots,m$, $\sigma(\Delta(\vec{r}))$ contains an `edge' with values $1,2,\ldots,p^{r_i}$. This is illustrated in Figures~\ref{fig:example2D} and \ref{fig:example3D}. This means that we can easily give a lower bound on the dimension increase when enlarging the code $C(L(\delta))$. To ease the notation in the following, we will order the exponents $r_i$ such that $r_1\geq r_2\geq \cdots\geq r_m$.
\begin{proposition}\label{prop:enlargement}
  Let $q=p^r$, and let $\vec{r}\in\mathbb{Z}_+^m$ be a vector such that $r_i\mid r$ for each $i$ and $r_1\geq r_2\geq\cdots\geq r_m$.
  Additionally, let $2< \delta\leq p^{r_2}+1$, and let $K$ be the largest index such that $\delta-1\leq p^{r_K}$. Then if $C(L(\delta))$ is a self-orthogonal $[n,k]$ code, there exists a quantum error-correcting code with parameters
  \begin{equation}\label{eq:enlargeParameters}
    [[n,\geq 2k-n+K,\geq\delta]]_q.
  \end{equation}
\end{proposition}
\begin{proof}
  Write $\mathcal{C}=C(L(\delta))$, and let $\mathcal{C}'=C(L(\delta-1))$. Since $1<\delta-1\leq p^{r_K}$, the observation at the start of this section implies that there are at least $K\geq 2$ monomials $X^{\vec{a}}\in\Delta(\vec{r})$ such that $\sigma(X^{\vec{a}})=\delta-1$. Thus, $\mathcal{C}'$ has dimension $k'\geq k+K$. As described in Section~\ref{subsec:improved-codes}, $\mathcal{C}$ and $\mathcal{C}'$ have minimal distances $\delta$ and $\delta-1$, respectively. Thus the observation below Theorem~\ref{thm:qArySteane} ensures that $d(\mathcal{C}',\mathcal{C}'^\perp)=d(\mathcal{C}')=\delta-1$, and we obtain
  \begin{equation*}
    \big\lceil(1+\textstyle\frac{1}{q})d(\mathcal{C}')\big\rceil
    = \big\lceil(1+\textstyle\frac{1}{q})(\delta-1)\big\rceil
    =\delta,
  \end{equation*}
  where the last equality stems from the assumption that $\delta-1\leq p^{r_2}\leq q$. The claim now follows by applying Theorem~\ref{thm:qArySteane} to $\mathcal{C}$ and $\mathcal{C}'$, and by using the bound $d(\mathcal{C},\mathcal{C}'^\perp)\geq d(\mathcal{C})=\delta$.
\end{proof}
A few additional remarks can be made about the Steane-enlargement described in Proposition~\ref{prop:dimAndDist}. First of all, the observation that leads to Proposition~\ref{prop:dimAndDist} does not help in the case $\delta> q+1$ since we require $\delta\leq p^{r_2}+1\leq q+1$. This does not mean that Steane-enlargement is impossible for $\delta>q+1$, but merely that we cannot \emph{guarantee} that enlargement is possible.

Secondly, the increase in dimension when applying Steane-enlargement to the code $C(L(\delta))$ may be greater than the $K$ specified in Proposition~\ref{prop:enlargement} since this $K$ is determined by considering monomials along the `edges' as in Figure~\ref{fig:example3D}. There may be several other monomials that have $\sigma$-value $\delta-1$, yielding a quantum error-correcting code with even better parameters. In Section~\ref{sec:exactIncrease}, we characterize the situations where this may happen, and give an improved bound in such cases. First, however, we illustrate the result through an example.

\begin{example}\label{ex:q9}
  Let $q=3^2=9$ and $\vec{r}=(2,2,1)$. The classical code $C(L(4))$ has parameters $[243,236,4]_9$, whence the CSS-construction, Theorem~\ref{thm:cssSelfOrth}, gives a $[[243,229,4]]_9$ quantum code. Since $\delta-1=3=p^{r_3}$, Proposition~\ref{prop:enlargement} ensures that Steane-enlargement will instead provide a quantum code with parameters $[[243,\geq 232,\geq 4]]_9$. In this case, the true dimension is in fact $232$.

  Using the same $q$ and $\vec{r}$, the code $C(L(7))$ is a $[243,221,7]_9$ classical code, yielding a $[[243,199,7]]_9$ quantum code via the CSS-construction. This time, Proposition~\ref{prop:enlargement} only guarantees a dimension increase of $2$ when applying Steane-enlargement, but the true parameters of the enlarged code are $[[243, 207, \geq 7]]_9$, meaning that the dimension has been increased by $8$.
\end{example}

\subsection{Determining the exact dimension increase}\label{sec:exactIncrease}
As previously mentioned, the dimension of an enlarged code may be greater than predicted in \eqref{eq:enlargeParameters}.
In this section, we will generalize the map $\tau^{(q)}$ from \cite{CG19} to provide an algorithm for computing the exact dimension increase when applying Steane-enlargement to the code $C(L(\delta))$. This generalization will also aid in characterizing those values of $\delta$ where Proposition~\ref{prop:enlargement} underestimates the dimension.
\begin{definition}\label{defi:tau}
  For $s\in\mathbb{Z}_+$ and $\vec{r}\in\mathbb{Z}_+^m$, we let $\tau^{(\vec{r})}(s)$ denote the number of tuples $(d_1,d_2,\ldots,d_m)$ such that $1\leq d_i\leq p^{r_i}$ for every $i$, and such that $s=\prod_{i=1}^md_i$.
\end{definition}
\begin{proposition}\label{prop:tauPrimeNonPrime}
  Let $s$ and $\vec{r}$ be as in Definition~\ref{defi:tau}, and assume that $r_1\geq r_2\geq\cdots\geq r_m$. Let $K$ be the largest index such that $s\leq p^{r_K}$. Then if $s$ is\textellipsis
  \begin{itemize}
    \item \textellipsis prime, we have $\tau^{(\vec{r})}(s)=K$
    \item \textellipsis square, we have $\tau^{(\vec{r})}(s)\geq K+\binom{K}{2}$
    \item \textellipsis non-prime and non-square, we have $\tau^{(\vec{r})}(s)\geq K^2$
  \end{itemize}
\end{proposition}
\begin{proof}
  Assume first that $s$ is prime. Then any tuple $(d_1,d_2,\ldots,d_m)\in\mathbb{Z}_+$ with $s=\prod_{i=1}^md_i$ must have $d_i=s$ for some $i$ and $d_j=1$ for $j\neq i$. Hence, in this case $\tau^{(\vec{r})}(s)$ is the number of indices $i$ such that $d_i\leq p^{r_i}$, which is exactly $K$.

  If $s$ is non-prime, there are still $K$ tuples with a single entry greater than $1$ as in the prime case. But we may also split $s$ in two factors $s=f_1f_2$ such that $f_1f_2<s\leq p^{r_K}$. Now, for any distinct indices $i_1,i_2\in\{1,2,\ldots,K\}$, the tuple $(d_1,d_2,\ldots,d_m)$ with $d_{i_1}=f_1$, $d_{i_2}=f_2$, and $d_i=1$ for $i\notin\{i_1,i_2\}$ is one of the tuples counted by $\tau^{(\vec{r})}(s)$. The number of ways to choose the indices $i_1,i_2$ is $K(K-1)$. If $s$ is not a square number, $f_1$ and $f_2$ are distinct, and each of the $K(K-1)$ choices of $i_1,i_2$ leads to a distinct tuple. Is $s$ is a square, we may have $f_1=f_2$, and the number of distinct tuples is instead $K(K-1)/2=\binom{K}{2}$. In both cases, we obtain the claimed inequality by adding $K$. 
\end{proof}

\begin{proposition}\label{prop:numDivisors}
  Let $s\in\mathbb{Z}_+$. Then the number of monomials $X^{\vec{a}}\in\Delta(\vec{r})$ that have $\sigma(X^{\vec{a}})=s$ is $\tau^{(\vec{r})}(s)$.
\end{proposition}
\begin{proof}
  We have $\sigma(X^{\vec{a}})=s$ if and only if $\prod_{i=1}^m(p^{r_i}-a_i)=s$. Since $0\leq a_i<p^{r_i}$, this is equivalent to $\prod_{i=1}^m d_i=s$ for $1\leq d_i\leq p^{r_i}$, proving the proposition.
\end{proof}
Combining Propositions~\ref{prop:tauPrimeNonPrime} and \ref{prop:numDivisors}, we obtain the following immediate corollary.
\begin{corollary}\label{cor:primeDimension}
  Let $q$, $\vec{r}$, and $\delta$ be as in Proposition~\ref{prop:enlargement}. Then \eqref{eq:enlargeParameters} gives the true dimension if and only if $\delta-1$ is a prime number.
  If $\delta-1$ is not a prime, the bound on the dimension may be increased by $\binom{K}{2}$ if $\delta-1$ is a square number and by $K(K-1)$ otherwise.
\end{corollary}

\begin{example}\label{ex:q9cont}
  We now return to the codes in Example~\ref{ex:q9}. In the case of $C(L(4))$, we saw that Proposition~\ref{prop:enlargement} gave the true minimal distance. Having established Corollary~\ref{cor:primeDimension}, we now know that this is no coincidence since $\delta-1=3$ is a prime number.

  For the code $C(L(7))$, $\delta-1=6$ is neither prime nor square. Consequently, Corollary~\ref{cor:primeDimension} tells us that the dimension must increase by at least $K^2=2^2=4$, which is $2$ more than the bound from Proposition~\ref{prop:enlargement}. Both bounds are, however, still smaller than the true value of $8$.
\end{example}
Since it may not be obvious how to compute $\tau^{(\vec{r})}$, we give the following recursive algorithm. Its correctness can be shown by a simple inductive argument.
\begin{algorithm}\label{alg:tau}
  On input $\vec{r}=(r_1,r_2,\ldots,r_m)$ and $s\in\mathbb{Z}_+$, this algorithm computes $\tau^{(\vec{r})}(s)$:
  \begin{enumerate}
    \item Check if $\vec{r}$ is a single value $r_1$. If this is the case, return $1$ if $s\leq r_1$, and $0$ otherwise.
    \item Initialize a counter variable $c:=0$
    \item For each integer $d\in\{1,2,\ldots,p^{r_1}\}$ with $d\mid s$, do the following:
    \begin{itemize}
      \item Let $\vec{r}'=(r_2,r_3,\ldots,r_m)$ and compute $\tau^{(\vec{r}')}(s/d)$.
      \item Update $c$ to be $c:=c+\tau^{(\vec{r}')}(s/d)$
    \end{itemize}
    \item Return $c$
  \end{enumerate}
\end{algorithm}
Since the number of $d$'s considered in Algorithm~\ref{alg:tau} is $\prod_{i=1}^{m-1}p^{r_i}=n/p^{r_m}$, the total number of operations is $\mathcal{O}(n/p^{r_m})$. This is a factor $p^{r_m}$ better than considering all $X^{\vec{a}}\in\Delta(\vec{r})$ and counting the ones with $\sigma(X^{\vec{a}})=s$. We collect these observations on Algorithm~\ref{alg:tau} and its relation to Proposition~\ref{prop:tauPrimeNonPrime} in the following proposition.
\begin{proposition}\label{prop:algorithm}
  Let $q,\vec{r},\delta$, and $K$ be as in Proposition~\ref{prop:enlargement}. Algorithm~\ref{alg:tau} correctly computes $\tau^{(\vec{r})}(\delta)=K$ in $\mathcal{O}(n/p^{r_m})$ operations, where $n=\prod_{i=1}^m p^{r_i}$.
\end{proposition}

\subsection{Examples of parameters}\label{sec:exampleParams}
To conclude our exposition, we give concrete parameters of Steane-enlarged codes in several examples. For each code presented here, we will compare it to the Gilbert-Varshamov bound from \cite{FengMa}.
\begin{theorem}\label{thm:gv}
  Let $n>k\geq 2$ with $n\equiv k\pmod{2}$, and let $d\geq 2$. Then there exists a pure stabilizer quantum code $[[n,k,d]]_q$ if the inequality
  \begin{equation}\label{eq:gvBound}
    \sum_{i=1}^{d-1}(q^2-1)^i\binom{n}{i}<q^{n-k+2}-1
  \end{equation}
  is satisfied.
\end{theorem}
In the same way as \cite{MTT16}, we will use the notation $[[n,k,d]]_q\gvMark$ in the following to indicate that the parameters $(n,k,d)$ exceed the Gilbert-Varshamov bound -- i.e. that \eqref{eq:gvBound} is not satisfied -- and we will write $[[n,k,d]]_q\gvMark*$ if $(n,k,d)$ satisfies \eqref{eq:gvBound}, but $(n,k,d+1)$ does not.
This is only possible for $n\equiv k\pmod{2}$, which is always the case for CSS-codes from self-orthogonal codes, but not necessarily for Steane-enlarged codes. Thus, for code parameters $(n,k,d)$ with $n\not\equiv k\pmod{2}$, we will use the same notation, albeit with the bound applied to the parameters $(n,k-1,d)$.
There is another bound, \cite[Cor.\,4.3]{JinXing}, which covers all values of $n$ and $k$. For the parameters presented in the current work, however, that bound is weaker than \eqref{eq:gvBound}, and several of the codes in the examples below exceed \cite[Cor.\,4.3]{JinXing} but not Theorem~\ref{thm:gv}. For this reason, we shall use Theorem~\ref{thm:gv} throughout.

In addition to the Gilbert-Varshamov bound, we will refer to the quantum Singleton bound in some cases. This bound is
\begin{equation}\label{eq:singleton}
  2d\leq n-k+2,
\end{equation}
and its proof can be found in \cite{KillLaflamme,rains}.
\begin{example}
  This is a continuation of Examples~\ref{ex:q9} and \ref{ex:q9cont}. When compared with the Gilbert-Varshamov bound, Theorem~\ref{thm:gv}, the CSS-code with parameters $[[243,229,4]]_9\gvMark*$ and the Steane-enlarged code with parameters $[[243,232,4]]_9\gvMark*$ meet the bound, whereas the two codes of minimal distance $7$ neither meet nor exceed the bound.
\end{example}

\begin{example}
  Consider $q=3^2=9$ and $\vec{r}=(2,1)$ as in Figure~\ref{fig:example2D}. Here, Proposition~\ref{prop:enlargement} guarantees that we can enlarge the CSS-codes $[[27, 21, 3]]_9\gvMark*$ and $[[27, 17, 4]]_9\gvMark*$ to codes of parameters $[[27, 23, \geq 3]]_9\gvMark$ and $[[27, 19, \geq 4]]_9\gvMark*$, respectively. Furthermore, Corollary~\ref{cor:primeDimension} ensures that these are the true dimensions. In fact, the code $[[27, 23, 3]]_9\gvMark$ is optimal since it meets the Singleton-bound \eqref{eq:singleton}.

  There are two additional Steane-enlarged codes that are not captured by Proposition~\ref{prop:enlargement}. These are $[[27, 13, 5]]_9$ enlarged to $[[27, 15, \geq 5]]_9\gvMark*$, and $[[27, 5, 7]]_9$ enlarged to $[[27, 8, \geq 7]]_9$, where the increases in dimension have been computed using Algorithm~\ref{alg:tau}. In both cases, the technique in Proposition~\ref{prop:enlargement} fails because $\delta>4=p^{r_2}+1$.
\end{example}

\begin{table}[b]
  \centering
  \footnotesize
  \begin{tabular}{ccccc}
    \toprule
    \multicolumn{2}{c}{\textbf{Construction}}&\multicolumn{3}{c}{\textbf{Dimension increase}}\\
    \cmidrule(lr){1-2}\cmidrule(lr){3-5}
\textbf{Thm.~\ref{thm:cssSelfOrth}} & \textbf{Thm.~\ref{thm:qArySteane}} & \textbf{Prop.~\ref{prop:enlargement}} & \textbf{Cor.~\ref{cor:primeDimension}} & \textbf{Prop.~\ref{prop:algorithm}}\\
    \midrule
    $[[64, 58, 3]]_8\gvMark*$	& $[[64, 60, 3]]_8\gvMark$ & $2$ & $2$\trueMark & $2$\\
    $[[64, 54, 4]]_8\gvMark*$	& $[[64, 56, 4]]_8\gvMark$ & $2$ & $2$\trueMark & $2$\\
    $[[64, 48, 5]]_8$	& $[[64, 51, 5]]_8\gvMark*$ & $2$ & $3$\phMark & $3$\\
    $[[64, 44, 6]]_8$	& $[[64, 46, 6]]_8\gvMark*$ & $2$ & $2$\trueMark & $2$\\
    $[[64, 36, 7]]_8$	& $[[64, 40, 7]]_8$ & $2$ & $4$\phMark & $4$\\
    $[[64, 32, 8]]_8$	& $[[64, 34, 8]]_8$ & $2$ & $2$\trueMark & $2$\\
    \bottomrule
  \end{tabular}
  \caption{Steane-enlarged codes with $q=2^3=8$ and $\vec{r}=(3,3)$. Additional information on how to read the table may be found in Example~\ref{ex:tables}.}
  \label{tab:q8}
\end{table}

\begin{table}[b]
  \centering
  \footnotesize
  \begin{tabular}{ccccc}
    \toprule
    \multicolumn{2}{c}{\textbf{Construction}}&\multicolumn{3}{c}{\textbf{Dimension increase}}\\
    \cmidrule(lr){1-2}\cmidrule(lr){3-5}
    \textbf{Thm.~\ref{thm:cssSelfOrth}} & \textbf{Thm.~\ref{thm:qArySteane}} & \textbf{Prop.~\ref{prop:enlargement}} & \textbf{Cor.~\ref{cor:primeDimension}} & \textbf{Prop.~\ref{prop:algorithm}}\\
    \midrule
    $[[625, 615, 3]]_5\gvMark*$ & $[[625, 619, 3]]_5\gvMark$ & $4$ & $4$\trueMark & $4$\\
    $[[625, 607, 4]]_5\gvMark*$ & $[[625, 611, 4]]_5\gvMark$ & $4$ & $4$\trueMark & $4$\\
    $[[625, 587, 5]]_5$ & $[[625, 597, 5]]_5$ & $4$ & $8$\phMark & $10$\\
    $[[625, 579, 6]]_5$ & $[[625, 583, 6]]_5$ & $4$ & $4$\trueMark & $4$\\
    \bottomrule
  \end{tabular}
  \caption{Steane-enlarged codes with $q=5$ and $\vec{r}=(1,1,1,1)$. Additional information on how to read the table may be found in Example~\ref{ex:tables}.}
  \label{tab:q5}
\end{table}

\begin{example}\label{ex:tables}
  In Tables~\ref{tab:q8}--\ref{tab:q64}, we list parameters of quantum codes in various cases where Proposition~\ref{prop:enlargement} guarantees that enlargement is possible. The tables contain both the original CSS-code and its Steane-enlarged code along with the predicted dimension increases from Proposition~\ref{prop:enlargement} and Corollary~\ref{cor:primeDimension}.
  
  In these tables, the first column shows the parameters of quantum codes obtained by applying Theorem~\ref{thm:cssSelfOrth} to self-orthogonal codes of the form $\mathcal{C}=C(L(\delta))$. The second column shows the results of enlarging the codes in the first column using $\mathcal{C}'=C(L(\delta-1))$ in Theorem~\ref{thm:qArySteane}. The third column gives the dimension increase guaranteed by Proposition~\ref{prop:enlargement}, and the fourth shows the bound provided by Corollary~\ref{cor:primeDimension}. Any number marked with an asterisk is \emph{known} to be the true value since $\delta-1$ is a prime. The final column shows the actual increase as computed by Algorithm~\ref{alg:tau}.

  The codes in Table~\ref{tab:q8} have better parameters than those in \cite[Tables\,1 and 2]{LaGuardiaAlves} for small values of $\delta$. More concretely, \cite{LaGuardiaAlves} lists codes with parameters $[[63,57,\geq 3]]_8\gvMark*$, $[[63,53,\geq 4]]_8\gvMark*$, $[[63,49,\geq 5]]_8\gvMark*$, and $[[63,45,\geq 6]]_8\gvMark*$. For larger values of $\delta$, however, \cite{LaGuardiaAlves} outperforms the codes in Table~\ref{tab:q8}. Likewise, for $\delta=3,4$ the parameters of the codes in Table~\ref{tab:q5} surpass those presented in \cite[Tables\,1 and 2]{LaguardiaPalazzo}. There, codes with parameters $[[624,614,\geq 3]]_5\gvMark*$ and $[[624,606,\geq 4]]_5\gvMark*$ are given. As before, \cite{LaguardiaPalazzo} also contains codes with higher minimal distances that have better parameters than the corresponding codes obtained in this work.
  All the codes in Tables~\ref{tab:q8} and \ref{tab:q5} have with $q=p$ and $r_i=1$, which are in fact special cases of hyperbolic codes. It seems to be a general pattern for such codes, that the Steane-enlargements with small distances outperform the codes in \cite{LaGuardiaAlves,LaguardiaPalazzo}, but that this relation is reversed for larger distances.

  Additionally, the codes in Table~\ref{tab:q8} have favourable parameters compared to the codes from \cite[Ex.\,5]{MTT16} defined from the Suzuki curve. Specifically, the codes in \cite{MTT16} have parameters $[[64,54,3]]_8$, $[[64,52,4]]_8\gvMark*$, $[[64,42,5]]_8$, $[[64,40,6]]_8$, $[[64,38,7]]_8$, and $[[64,36,8]]_8$, which are all worse than those in Table~\ref{tab:q8} except the one with distance $8$. As a final remark, the code with parameters $[[64, 60, 3]]_8\gvMark$ meets the quantum Singleton bound \eqref{eq:singleton}.

  The codes in Tables~\ref{tab:q16} and \ref{tab:q64} have parameters that cannot be achieved using the method from \cite{LaGuardiaAlves,LaguardiaPalazzo} since those codes all have lengths $q^m-1$ for some $m\geq 2$, where $q$ is the field size. Studying the tables, it is also evident that Corollary~\ref{cor:primeDimension} provides a better bound for the dimension than Proposition~\ref{prop:enlargement}, but that the actual increase in dimension may be significantly higher. In any case, however, Proposition~\ref{prop:algorithm} ensures that the true increase can be computed using Algorithm~\ref{alg:tau}.

  Among the codes presented here, two were MDS-codes: $[[27, 23, 3]]_9\gvMark$ and $[[64, 60, 3]]_8\gvMark$.
  From recent work \cite[Cor.\,3.10]{LMS19} the same lengths, dimensions, and minimal distances can be achieved, but the field size is much larger. In particular, they require $q>n$ so the corresponding field sizes are at least $29$ and $67$, respectively.
\end{example}

\begin{table}
  \centering
  \footnotesize
  \begin{tabular}{ccccc}
    \toprule
    \multicolumn{2}{c}{\textbf{Construction}}&\multicolumn{3}{c}{\textbf{Dimension increase}}\\
    \cmidrule(lr){1-2}\cmidrule(lr){3-5}
    \textbf{Thm.~\ref{thm:cssSelfOrth}} & \textbf{Thm.~\ref{thm:qArySteane}} & \textbf{Prop.~\ref{prop:enlargement}} & \textbf{Cor.~\ref{cor:primeDimension}} & \textbf{Prop.~\ref{prop:algorithm}}\\
    \midrule
    $[[1024, 1016, \ph3]]_{16}\gvMark*$   & $[[1024, 1019, \ph3]]_{16}\gvMark$   & $3$ & $3$\trueMark & $3$  \\
    $[[1024, 1010, \ph4]]_{16}\gvMark*$   & $[[1024, 1013, \ph4]]_{16}\gvMark*$   & $3$ & $3$\trueMark & $3$  \\
    $[[1024, \ph998, \ph5]]_{16}$ & $[[1024, 1004, \ph5]]_{16}$   & $3$ & $6$\phMark & $6$  \\
    $[[1024, \ph994, \ph6]]_{16}$ & $[[1024, \ph996, \ph6]]_{16}$ & $2$ & $2$\trueMark & $2$  \\
    $[[1024, \ph978, \ph7]]_{16}$ & $[[1024, \ph986, \ph7]]_{16}$ & $2$ & $4$\phMark & $8$  \\
    $[[1024, \ph974, \ph8]]_{16}$ & $[[1024, \ph976, \ph8]]_{16}$ & $2$ & $2$\trueMark & $2$  \\
    $[[1024, \ph956, \ph9]]_{16}$ & $[[1024, \ph965, \ph9]]_{16}$ & $2$ & $4$\phMark & $9$  \\
    $[[1024, \ph946, 10]]_{16}$   & $[[1024, \ph951, 10]]_{16}$   & $2$ & $3$\phMark & $5$  \\
    $[[1024, \ph934, 11]]_{16}$   & $[[1024, \ph940, 11]]_{16}$   & $2$ & $4$\phMark & $6$  \\
    $[[1024, \ph930, 12]]_{16}$   & $[[1024, \ph932, 12]]_{16}$   & $2$ & $2$\trueMark & $2$  \\
    $[[1024, \ph900, 13]]_{16}$   & $[[1024, \ph915, 13]]_{16}$   & $2$ & $4$\phMark & $15$ \\
    $[[1024, \ph896, 14]]_{16}$   & $[[1024, \ph898, 14]]_{16}$   & $2$ & $2$\trueMark & $2$  \\
    $[[1024, \ph884, 15]]_{16}$   & $[[1024, \ph890, 15]]_{16}$   & $2$ & $4$\phMark & $6$  \\
    $[[1024, \ph872, 16]]_{16}$   & $[[1024, \ph878, 16]]_{16}$   & $2$ & $4$\phMark & $6$  \\
    $[[1024, \ph848, 17]]_{16}$   & $[[1024, \ph860, 17]]_{16}$   & $2$ & $3$\phMark & $12$ \\
    \bottomrule
  \end{tabular}
  \caption{Steane-enlarged codes with $q=2^4=16$ and $\vec{r}=(4,4,2)$. Additional information on how to read the table may be found in Example~\ref{ex:tables}.}
  \label{tab:q16}
\end{table}

\begin{table}
  \centering
  \footnotesize
  \begin{tabular}{ccccc}
    \toprule
\multicolumn{2}{c}{\textbf{Construction}}&\multicolumn{3}{c}{\textbf{Dimension increase}}\\
    \cmidrule(lr){1-2}\cmidrule(lr){3-5}
    \textbf{Thm.~\ref{thm:cssSelfOrth}} & \textbf{Thm.~\ref{thm:qArySteane}} & \textbf{Prop.~\ref{prop:enlargement}} & \textbf{Cor.~\ref{cor:primeDimension}} & \textbf{Prop.~\ref{prop:algorithm}}\\
    \midrule
    $[[1024, 1014, 3]]_8\gvMark*$ & $[[1024, 1018, 3]]_8\gvMark$  & $4$ & $4$\trueMark & $4$  \\
    $[[1024, 1008, 4]]_8\gvMark*$ & $[[1024, 1011, 4]]_8\gvMark*$ & $3$ & $3$\trueMark & $3$  \\
    $[[1024, \ph990, 5]]_8$       & $[[1024, \ph999, 5]]_8$       & $3$ & $6$\phMark   & $9$  \\
    $[[1024, \ph984, 6]]_8$       & $[[1024, \ph987, 6]]_8$       & $3$ & $3$\trueMark & $3$  \\
    $[[1024, \ph960, 7]]_8$       & $[[1024, \ph972, 7]]_8$       & $3$ & $9$\phMark   & $12$ \\
    $[[1024, \ph954, 8]]_8$       & $[[1024, \ph957, 8]]_8$       & $3$ & $3$\trueMark & $3$  \\
    $[[1024, \ph922, 9]]_8$       & $[[1024, \ph938, 9]]_8$       & $3$ & $9$\phMark   & $16$ \\
    \bottomrule 
  \end{tabular}
  \caption{Steane-enlarged codes with $q=2^3=8$ and $\vec{r}=(3,3,3,1)$. Additional information on how to read the table may be found in Example~\ref{ex:tables}.}
  \label{tab:q64}
\end{table}

\section*{Acknowledgements}
The authors express their gratitude to Diego Ruano for delightful discussions in relation to this work.

\section*{References}\renewcommand{\chapter}[2]{\footnotesize}


\begin{thebibliography}{10}
\providecommand{\url}[1]{{#1}}
\providecommand{\urlprefix}{URL }
\providecommand{\doi}[1]{DOI~\discretionary{}{}{}\href{https://doi.org/#1}{#1}}

\bibitem{CalderbankShor}
Calderbank, A.R., Shor, P.W.: Good quantum error-correcting codes exist.
\newblock Phys. Rev. A \textbf{54}, 1098--1105 (1996).
\newblock \doi{10.1103/PhysRevA.54.1098}

\bibitem{CG19}
Christensen, R.B., Geil, O.: Steane-enlargement of quantum codes from the
  {H}ermitian curve.
\newblock CoRR \textbf{abs/1904.10007} (2019).
\newblock \urlprefix\url{http://arxiv.org/abs/1904.10007}

\bibitem{FengMa}
Feng, K., Ma, Z.: A finite {G}ilbert-{V}arshamov bound for pure stabilizer
  quantum codes.
\newblock IEEE Trans. Inf. Theory \textbf{50}(12), 3323--3325 (2004).
\newblock \doi{10.1109/TIT.2004.838088}

\bibitem{GGHR18}
Galindo, C., Geil, O., Hernando, F., Ruano, D.: Improved constructions of
  nested code pairs.
\newblock {IEEE} Trans. Inf. Theory \textbf{64}(4), 2444--2459 (2018).
\newblock \doi{10.1109/TIT.2017.2755682}

\bibitem{GHR15}
Galindo, C., Hernando, F., Ruano, D.: Stabilizer quantum codes from {J}-affine
  variety codes and a new {S}teane-like enlargement.
\newblock Quantum Inf. Process. \textbf{14}(9), 3211--3231 (2015).
\newblock \doi{10.1007/s11128-015-1057-2}

\bibitem{GH01}
Geil, O., H{\o}holdt, T.: On hyperbolic codes.
\newblock In: Applied Algebra, Algebraic Algorithms and Error-Correcting Codes,
  14th International Symposium, AAECC-14, Melbourne, Australia November 26-30,
  2001, Proceedings, pp. 159--171 (2001).
\newblock \doi{10.1007/3-540-45624-4\_17}

\bibitem{LaGuardiaAlves}
Guardia, G.G.L., Alves, M.M.: On cyclotomic cosets and code constructions.
\newblock Linear Algebra Appl. \textbf{488}, 302 -- 319 (2016).
\newblock \doi{10.1016/j.laa.2015.09.034}

\bibitem{LaguardiaPalazzo}
Guardia, G.G.L., Palazzo, R.: Constructions of new families of nonbinary {CSS}
  codes.
\newblock Discrete Mathematics \textbf{310}(21), 2935 -- 2945 (2010).
\newblock \doi{10.1016/j.disc.2010.06.043}

\bibitem{Hamada}
Hamada, M.: Concatenated quantum codes constructible in polynomial time:
  Efficient decoding and error correction.
\newblock {IEEE} Trans. Inf. Theory \textbf{54}(12), 5689--5704 (2008).
\newblock \doi{10.1109/TIT.2008.2006416}

\bibitem{JinXing}
{Jin}, L., {Xing}, C.: Quantum {G}ilbert-{V}arshamov bound through symplectic
  self-orthogonal codes.
\newblock In: 2011 IEEE International Symposium on Information Theory
  Proceedings, pp. 455--458 (2011).
\newblock \doi{10.1109/ISIT.2011.6034167}

\bibitem{KillLaflamme}
Knill, E., Laflamme, R.: Theory of quantum error-correcting codes.
\newblock Phys. Rev. A \textbf{55}, 900--911 (1997).
\newblock \doi{10.1103/PhysRevA.55.900}

\bibitem{qArySteane}
Ling, S., Luo, J., Xing, C.: Generalization of {S}teane's enlargement
  construction of quantum codes and applications.
\newblock {IEEE} Trans. Inf. Theory \textbf{56}(8), 4080--4084 (2010).
\newblock \doi{10.1109/TIT.2010.2050828}

\bibitem{LMS19}
L{\'{o}}pez, H.H., Matthews, G.L., Soprunov, I.: Monomial-{C}artesian codes and
  their duals, with applications to {LCD} codes, quantum codes, and locally
  recoverable codes.
\newblock CoRR \textbf{abs/1907.11812} (2019).
\newblock \urlprefix\url{http://arxiv.org/abs/1907.11812}

\bibitem{MTT16}
Munuera, C., Ten{\'o}rio, W., Torres, F.: Quantum error-correcting codes from
  algebraic geometry codes of {C}astle type.
\newblock Quantum Inf. Process. \textbf{15}(10), 4071--4088 (2016).
\newblock \doi{10.1007/s11128-016-1378-9}

\bibitem{rains}
{Rains}, E.M.: Nonbinary quantum codes.
\newblock IEEE Transactions on Information Theory \textbf{45}(6), 1827--1832
  (1999).
\newblock \doi{10.1109/18.782103}

\bibitem{Steane96}
Steane, A.: Multiple-particle interference and quantum error correction.
\newblock Proc. R. Soc. Lond. Ser. A: Math. Phys. Eng. Sci. \textbf{452}(1954),
  2551--2577 (1996)

\bibitem{Steane99}
Steane, A.M.: Enlargement of {C}alderbank-{S}hor-{S}teane quantum codes.
\newblock {IEEE} Trans. Inf. Theory \textbf{45}(7), 2492--2495 (1999).
\newblock \doi{10.1109/18.796388}

\end{thebibliography}
\end{document}